\documentclass[number]{llncs}
\usepackage{amsmath}
\usepackage{amssymb}
\usepackage{tikz}
\usepackage[strings]{underscore}
\usepackage{url}
\usetikzlibrary{calc}

\newcommand{\ama}{\[ \begin{aligned}}
\newcommand{\ema}{\end{aligned} \]}
\newcommand{\set}[1]{\left\{#1\right\}}

\newcommand{\R}{\mathbb{R}}

\def\eg{{e.\,g.}}
\def\ie{{i.\,e.}}

\begin{document}


\title{A tight bound on the speed-up through storage for quickest multi-commodity flows\thanks{Supported by the DFG Priority Programme 1307 ``Algorithm Engineering'' (grant SK~58/6-3) and the DFG Research Center \textsc{Matheon} ``Mathematics for key technologies'' in Berlin.} }

\titlerunning{A tight bound on the speed-up through storage for quickest multi-commodity flows}

\author{Martin Gro\ss{}\inst{1} \and Martin Skutella\inst{1}} 

\institute{Fakult\"at II -- Mathematik und Naturwissenschaften\\ Institut f\"ur Mathematik, Sekr.~MA 5-2\\Technische Universit\"at Berlin, Stra{\ss}e des 17.~Juni 136\\ 10623 Berlin, Germany}


\maketitle

\begin{abstract}
Multi-commodity flows over time exhibit the non-intuitive property that letting flow wait can allow us to send flow faster overall. Fleischer and Skutella (IPCO~2002) show that the speed-up through storage is at most a factor of~$2$, and that there are instances where the speed-up is as large as a factor of~$4/3$. We close this gap by presenting a family of instances for which the speed-up factor through storage converges to~$2$.\\

\end{abstract}

 
\section{Introduction}
 
Multi-commodity flows over time are a useful tool for modelling problems that involve non-instantaneous transportation through some kind of network. Some recent examples for such problems include aircraft routing~\cite{RoyThomlin07}, cloud computing~\cite{AgapiEtAl09,SoudanPrimet09}, road traffic networks~\cite{KoehlerMoehringSkut09}, disaster management~\cite{HanGuanShi11,OsmanEtAl09}, evacuations~\cite{LahmarAssavapokeeArdekani06}, logistics~\cite{AndreattaEtAl09,HernandezPeetaKalafatas11}, packet routing~\cite{PeisWiese11} and shared-ride systems~\cite{BraunWinter09}.
 
During transportation from a source to a destination, depending on the particular application, making stops before reaching the destination might be feasible or not. For example, if the underlying network is a street network, we can only stop if parking space is available. Similarly, storage of data is possible in telecommunication networks only if buffers are present. Since keeping buffers or parking spaces available costs money, the question arises how much worse the situation becomes if we forbid storage compared to the setting where storage is allowed. We will refer to making stops during transportation as \emph{storage} -- this is sometimes referred to as \emph{holdover} or \emph{waiting} as well.

\section{Quickest multi-commodity flows}
 
We start by defining the \emph{quickest multi-commodity flow problem}. For a more detailed introduction to the area of flows over time as well as for further pointers to the literature, we refer to~\cite{Skutella09,PhDThesisGross14}.

We consider a network given by a directed graph~$G=(V,A)$ with a set of nodes~$V$ and a set of arcs~$A$. Each arc~$a \in A$ has a capacity~$u_a\geq0$ specifying the maximum rate at which flow may enter it at any point in time, and a transit time~$\tau_a\geq0$ indicating how long flow needs to traverse from its start to its end node. Furthermore, we are given a set of~$k$ commodities $K=\set{0,1,\dots,k-1}$ with source nodes~$s_i\in V$, sink nodes~$t_i\in V$, and demands~$d_i\geq0$ for every commodity~$i\in K$. 
 
We are interested in finding a quickest flow, that is, a flow over time that fulfills all supplies and demands as quickly as possible, \ie, within minimum time horizon~$T$. More formally, a \emph{multi-commodity flow over time} $f$ with \emph{time horizon~$T>0$} is a function $f: (A \times K) \to ([0,T) \to \R_{\geq0})$ that assigns a Lebesque-integrable flow rate function $f^i_a: [0,T) \to \R_{\geq0}$ to every arc~$a\in A$ and every commodity~$i \in K$, subject to \emph{capacity} and \emph{flow conservation} constraints that we specify below.
 
The value~$f^i_a(\xi)$ determines the \emph{flow rate} of commodity~$i$ into arc~$a$ at time~$\xi$. The flow rate specifies, as its name indicates, the rate at which the flow is entering the arc. This could for example be the number of cars entering a road per time unit at a given point in time. Due to the transit times, flow entering an arc~$a = (v,w)$ at time $\xi$ will arrive at node~$w$ only at time $\xi+\tau_a$. While these flow rate functions are only defined on~$[0,T)$, for the sake of convenience we set~$f^i_a(\theta) := 0$ for~$\theta\in\R\setminus[0,T)$, $i\in K$, $a\in A$. 

We can now specify the constraints on the flow rates. \emph{Capacity constraints} make sure that the flow rates assigned to the commodities never violate the arc capacities:
\begin{align}\label{eq:capacities}
\sum_{i\in K} f^i_a(\theta) \leq u_a\qquad \text{for all $a \in A$, $\theta\in[0,T)$.}
\end{align}
\emph{Flow conservation constraints} ensure that flow of a commodity can only emerge from its corresponding source node:
\begin{align}\label{eq:weak-flow-conservation}
\begin{split}
\sum_{a = (\cdot,v) \in A} \int_0^{\theta-\tau_a} f^i_a(\xi)\ d\xi  - \!\!\!\!\sum_{a = (v,\cdot) \in A} &\int_0^\theta f^i_a(\xi)\ d\xi \geq 0\\
&\text{for all $i\in K$, $v\in V\setminus\{s_i\}$, $\theta\in[0,T)$.}	
\end{split}
\end{align}
Finally, in the end, at time~$T$, all flow must have reached its corresponding sink node:
\begin{align}\label{eq:fulfill-demands}
\begin{split}
\sum_{a = (\cdot,v) \in A} \int_0^{T-\tau_a} f^i_a(\xi)\ d\xi  - \!\!\!\!\sum_{a = (v,\cdot) \in A} \int_0^T f^i_a(\xi)\ d\xi =&
\begin{cases}
d_i & \text{if $v=t_i$,}\\
-d_i & \text{if $v=s_i$,}\\
0 & \text{otherwise,}	
\end{cases}\\
&\qquad\qquad\text{for all $i\in K$.}
\end{split}
\end{align}
A feasible solution to the quickest multi-commodity flow problem is a flow over time~$f$ with time horizon~$T$ that satisfies constraints~\eqref{eq:capacities}, \eqref{eq:weak-flow-conservation}, and~\eqref{eq:fulfill-demands}. The objective is to minimize the time horizon~$T$.

While for the case of a single commodity the quickest flow problem can be solved efficiently (cf.~\cite{FordFulkerson58,BurkardDlaskaKlinz93}), the multi-commodity case is weakly NP-hard~\cite{HallHipplerSkutella07}. On the positive side, there is a fully polynomial time approximation scheme (FPTAS)~\cite{FleischerSkutella02,FleischerSkutella07}.

Notice that the flow conservation constraints~\eqref{eq:weak-flow-conservation} allow flow to arrive at an intermediate node~$v$, wait there for some time, and continue onwards (\ie, the term on the left hand side of~\eqref{eq:weak-flow-conservation} might be positive at times). If storage of flow is prohibited at intermediate nodes, we can simply enforce equality in~\eqref{eq:weak-flow-conservation} for~$v\in V\setminus\{s_i,t_i\}$ at all times~$\theta$. Flows over time obeying these \emph{strict flow conservation} constraints are called \emph{flows over time without (intermediate) storage}. 

For the case of a single commodity, it follows from the work of Ford and Fulkerson~\cite{FordFulkerson58} that there is always a quickest flow that does not make use of intermediate storage. On the other hand, Fleischer and Skutella~\cite{FleischerSkutella02,FleischerSkutella07} present an instance of the quickest multi-commodity flow problem where prohibiting intermediate storage increases the optimal time horizon by a factor of~$4/3$. On the positive side, they show that, for any instance, the optimal time horizon without storage is at most a factor of~$2$ larger than the optimal time horizon with storage. Since then it has been an open problem to close this gap (see, \eg, Exercise~10 in~\cite{Skutella09}). In the next section we present a family of instances for which the ratio between the optimal time horizon without and with storage, respectively, converges to~$2$. That is, the speed-up through storage can be as large as a factor of~$2$ and this bound is tight.

\section{Instances with a large speed-up through storage}

Before we give a formal definition of our family of instances, we first provide an intuitive description; see also Figure~\ref{figure:network}. Imagine a single lane roundabout traffic packed with cars where each car wants to make one full turn on the roundabout before leaving it. In an ideal world, these cars drive at uniform speed without interfering with one another and complete their turns simultaneously, then leave the roundabout after, shall we say, 10~seconds. Now we add another car to this scenario that is waiting to enter the roundabout at some point~$P$ in order to make a full turn as well. If the other cars are not willing to stop and wait, the new car has to wait until everyone else is done (as there are no gaps between cars) and only then take its turn to finish after 20 seconds. If however, every car is willing to shortly leave the roundabout at~$P$ in order to re-enter after only a short stop, the new car can join immediately and everyone is done after only slightly more than 10~seconds.
\definecolor{DarkGray}    {gray}{0.2}
\definecolor{ForestGreen} {cmyk}{0.91,0,0.88,0.12}
\tikzstyle{node}=[draw, circle, color=black, fill=white, minimum size = 0.75cm]
\tikzstyle{invisible}=[color=white, fill=white, text=black]
\tikzstyle{arc}=[ultra thick,->,>=stealth]

\newcommand{\street}[3]{ \draw[DarkGray,line width=6mm] (#1:#3) -- (#1:#2);
    \draw[white,line width=5.5mm] (#1:#3) -- (#1:#2) -- ++(#1:-0.025);
    \draw[DarkGray,line width=5.0mm] (#1:#3) -- (#1:#2) -- ++(#1:-0.05);
    \draw[white,thick,dashed] (#1:#3) -- (#1:#2) -- ++(#1:-0.05);
    }

    \newcommand{\car}{
     \fill[very thin,fill=red,draw=black] (0.04,1.8) rectangle (0.225,2.2);
     \draw[very thin,draw=black,fill=blue] (0.06,1.9) to[out=-20,in=200] (0.205,1.9) -- (0.195,1.95) -- (0.07,1.95) -- (0.06,1.9);
     \draw[very thin,draw=black,fill=blue] (0.06,2.1) to[out=20,in=160] (0.205,2.1) -- (0.195,2.03) -- (0.07,2.03) -- (0.06,2.1);
     \draw[very thin,fill=yellow] (0.05,2.2) -- (0.09,2.2) -- (0.08,2.17) -- (0.06,2.17) -- (0.05,2.2);
     \draw[very thin,fill=yellow] (0.215,2.2) -- (0.175,2.2) -- (0.185,2.17) -- (0.205,2.17) -- (0.215,2.2);
    }
    
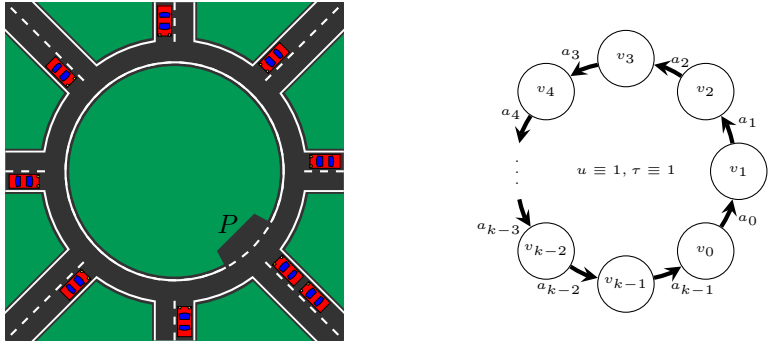
\begin{figure}[tb]
  \centering
  \begin{tikzpicture}[every node/.style={font=\tiny,inner sep=0mm}]
   \begin{scope}[xshift=-6cm,yshift=0cm]
    \fill[ForestGreen] (-2.25,-2.25) rectangle (2.25,2.25);
    \fill[DarkGray] (0,0) circle (1.8cm); 
    \draw[white,thick] (0,0) circle (1.75cm); 
    \draw[white,thick] (0,0) circle (1.45cm);
    \fill[ForestGreen] (0,0) circle (1.4cm); 
    \fill[DarkGray] (-28:1.6) -- (-62:1.6) -- (-62:1.2) -- (-28:1.2) -- (-28:1.6);
    \draw[white,thick,dashed] (0,0) circle (1.45cm);
    \node[] (a) at (-45:1) {\normalsize $P$};
    \begin{scope}[]
     \clip (-2.25,-2.25) rectangle (2.25,2.25);
     \street{0}{1.75}{2.25}
     \street{45}{1.75}{3.5}
     \street{90}{1.75}{2.25}
     \street{135}{1.75}{3.5}
     \street{180}{1.75}{2.25}
     \street{225}{1.75}{3.5}
     \street{270}{1.75}{2.25}
     \street{315}{1.75}{3.5}
     \end{scope}
     \begin{scope}[rotate around={0:(0,0)},yshift=-4cm]
      \car
     \end{scope}     
     \begin{scope}[rotate around={45:(0,0)},yshift=-4cm]
      \car
     \end{scope}
     \begin{scope}[rotate around={90:(0,0)},yshift=-4cm]
      \car
     \end{scope}     
     \begin{scope}[rotate around={135:(0,0)},yshift=-4cm]
      \car
     \end{scope}  
     \begin{scope}[rotate around={180:(0,0)},yshift=-4cm]
      \car
     \end{scope}  
     \begin{scope}[rotate around={225:(0,0)},yshift=-4cm]
      \car
     \end{scope}  
     \begin{scope}[rotate around={270:(0,0)},yshift=-4cm]
      \car
     \end{scope}       
     \begin{scope}[rotate around={315:(0,0)},yshift=-4cm]
      \car
     \end{scope}       
     \begin{scope}[rotate around={45:(0,0)},yshift=-4.5cm]
      \car
     \end{scope}       
   \end{scope}

   \def \radius {1.5}
   \node[node] (v1)           at (315:\radius) {$v_0$};
   \node[node] (v2)           at (0:\radius) {$v_1$};
   \node[node] (v3)           at (  45:\radius) {$v_2$};
   \node[node] (v4)           at (  90:\radius) {$v_3$};
   \node[node] (v5)           at (  135:\radius) {$v_4$};
	 \node[node,invisible] (v6) at (180:\radius) {}; \node[] (text) at ($(176.5:\radius) + (180:-0.05)$) {$\vdots$};
	 \node[node] (v7)           at (225:\radius) {$v_{k-2}$};
	 \node[node] (v8)           at (270:\radius) {$v_{k-1}$};
	 \node[] (text)             at (0:0) {$u \equiv 1, \tau \equiv 1$};
	 \path (v1) edge[arc, bend right=10] node[auto,swap,xshift=1mm] {$a_0$} (v2); 
	 \path (v2) edge[arc, bend right=10] node[auto,swap,xshift=1mm] {$a_1$} (v3); 
	 \path (v3) edge[arc, bend right=10] node[auto,swap] {$a_2$} (v4); 
	 \path (v4) edge[arc, bend right=10] node[auto,swap,yshift=1mm] {$a_3$} (v5); 
	 \path (v5) edge[arc, bend right=10] node[auto,swap,yshift=1mm] {$a_4$} (v6); 
	 \path (v6) edge[arc, bend right=10] node[auto,swap,yshift=-1mm] {$a_{k-3}$} (v7); 
	 \path (v7) edge[arc, bend right=10] node[auto,swap,yshift=-1mm] {$a_{k-2}$} (v8); 
	 \path (v8) edge[arc, bend right=10] node[auto,swap,yshift=-1mm] {$a_{k-1}$} (v1);
	\end{tikzpicture}
  \caption{A single lane roundabout traffic (left) and the corresponding network~$G$ (right). All arcs have unit capacities and transit times.\label{figure:network}}
 \end{figure}

Let $G=(V,A)$ be a directed cycle with~$k$ nodes~$V=\set{v_0,v_1,\dots,v_{k-1}}$ and arcs~$a_j=(v_j,v_{(j+1)\bmod k})$, $j=0,\dots,k-1$. All arcs have unit capacities and transit times. There are $k$ commodities~$K=\{0,1,\dots,k-1\}$ with~$s_i=v_i$ and~$t_i=v_{(i-1)\bmod k}$, for~$i=0,\dots,k-1$. Commodity~$0$ has demand~$d_0:=2$ and all remaining commodities have unit demand~$d_i:=1$, $i=1,\dots,k-1$.

 %


\begin{lemma}\label{lemma:lowerbound}
If storage of flow at intermediate nodes is permitted, all demands can be sent within time horizon~$k+1$.
\end{lemma}

\begin{proof}
The following multi-commodity flow over time~$f$ sends all demands within time horizon~$k+1$. 
\begin{itemize}
\item For commodity~$0$, we send flow at rate~$1$ into the path $s_0=v_0\to v_1\to\dots\to v_{k-1}=t_0$ during the time interval $[0,2)$. This flow does not wait at any node and thus arrives at its sink~$t_0$ before time $2+(k-1)=k+1$.
\item For commodities $i=1,\dots,k-1$, we send flow at rate~$1$ into the path $s_i=v_i\to\dots\to v_{k-1}\to v_0\to\dots\to v_{i-1}=t_i$ during the time interval $[0,1)$. For~$i>1$, flow of commodity~$i$ only waits at node~$v_0$ for one time unit before proceeding towards~$v_1$. These flows thus arrive at their sinks before time $1+(k-1)+1=k+1$.
\end{itemize}	
By definition, the flow over time~$f$ sends all demands and satisfies flow conservation constraints. Thus, we only need to check the capacity constraints. By our choice of flow rates, capacity violations can only occur if two commodities send flow into an arc at the same time. This is prevented by the waiting of commodities $2,\dots,k-1$ in $v_0$, thus completing the proof.
\end{proof}
 
It remains to lower bound the time horizon necessary to send all demands without storage. 

\begin{lemma}
\label{lemma:upperbound}
If storage of flow at intermediate nodes is forbidden, the minimal time horizon necessary to send all demands is larger than~$2k-2$.
\end{lemma}

\begin{proof}
We assume by contradiction that there is a flow over time~$f$ with time horizon~$T=2k-2$ that routes all demands without storage at intermediate nodes. Since the flow paths of all commodities have transit time~$k-1$, flow of any commodity must leave its source before time~$k-1$ and cannot reach its sink before time~$k-1$. In particular, flow cannot enter the last arc on its path before time~$k-2$. As a consequence, since flow may not wait at intermediate nodes and arc transit times are unit, all flow must be sent into some arc on its path within the time interval~$[k-2,k-1)$, that is,
\begin{align*}
\sum_{j=0}^{k-1}\sum_{i\in K}\int_{k-2}^{k-1}f^i_{a_j}(\xi)\ d\xi \geq \sum_{i\in K}d_i = k+1\enspace.
\end{align*} 
On the other hand, by capacity constraints, $\sum_{i\in K}f^i_{a_j}(\xi)\leq u_{a_j}=1$ such that
\begin{align*}
\sum_{j=0}^{k-1}\sum_{i\in K}\int_{k-2}^{k-1}f^i_{a_j}(\xi)\ d\xi \leq k\enspace.
\end{align*}
This contradiction concludes the proof.
\end{proof}
 
Together, Lemma~\ref{lemma:lowerbound} and Lemma~\ref{lemma:upperbound} yield the following theorem.
 
\begin{theorem}
There is a family of quickest multi-commodity flow instances for which the ratio between the optimal time horizon without and with intermediate storage, respectively, converges to~$2$.
\end{theorem}

We mention another interesting consequence of the family of instances discussed above. If we reduce the demand of commodity~$0$ to~$d_0:=1$, then there is a quickest flow with time horizon~$k$ which does not use storage at intermediate nodes. However, as soon as we increase~$d_0$ to~$1+\varepsilon$ for any~$\varepsilon>0$, the same arguments as in the proof of Lemma~\ref{lemma:upperbound} show that, if storage is prohibited, the optimal time horizon jumps above~$2k-2$. Thus, in contrast to the setting where storage is permitted (see Lemma~4.8 in~\cite{FleischerSkutella07}), the optimal time horizon is no longer a continuous function of the positive demand values if storage is prohibited.

\bibliographystyle{elsarticle-num} 
\bibliography{literature} 
\end{document}